\documentclass[sts]{imsart}

\RequirePackage{amsthm,amsmath,amsfonts,amssymb}
\RequirePackage[numbers]{natbib}
\RequirePackage[colorlinks,citecolor=blue,urlcolor=blue]{hyperref}
\RequirePackage{graphicx}

\startlocaldefs

\usepackage{lmodern}
\usepackage{enumerate}

\usepackage{tcolorbox}
\newtcolorbox[auto counter]{mybox}[2][]{float,title={\textcolor{black}{Box~\thetcbcounter: #2}},#1, 
colframe=white!70!gray}

\usepackage{cleveref}
\usepackage{notation}

\renewcommand*\citet[1]{\cite{#1}}

\crefname{figure}{Figure}{Figures}
\crefname{equation}{}{}
\crefname{definition}{Definition}{Definitions}
\crefname{corollary}{Corollary}{Corollaries}
\crefname{proposition}{Proposition}{Propositions}
\crefname{theorem}{Theorem}{Theorems}
\crefname{remark}{Remark}{Remarks}
\crefname{principle}{Principle}{Principles}
\crefname{lemma}{Lemma}{Lemmata}
\crefname{claim}{Claim}{Claims}
\crefname{table}{Table}{Tables}
\crefname{section}{Section}{Sections}
\crefname{subsection}{Subsection}{Subsections}
\crefname{subsubsection}{Subsection}{Subsection}
\crefname{assumption}{Assumption}{Assumptions}
\crefname{appendix}{Appendix}{Appendices}
\numberwithin{equation}{section}

\endlocaldefs

\begin{document}

\begin{frontmatter}
\title{Infer-and-widen, or not?}

\begin{aug}
\author[A]{\fnms{Ronan}~\snm{Perry}\ead[label=e1]{rflperry@uw.edu}},
\author[B]{\fnms{Zichun}~\snm{Xu}\ead[label=e2]{zx87@uw.edu}},
\author[C]{\fnms{Olivia}~\snm{McGough}\ead[label=e3]{mcgougho@uw.edu}},
\and
\author[D]{\fnms{Daniela}~\snm{Witten}\ead[label=e4]{dwitten@uw.edu}}
\address[A]{Ronan Perry is a PhD student in the Department of Statistics, University of Washington\printead[presep={\ }]{e1}.}
\address[B]{Zichun Xu is is a PhD student in the Department of Biostatistics, University of Washington\printead[presep={\ }]{e2}.}
\address[C]{Olivia McGough is a PhD student in the Department of Statistics, University of Washington\printead[presep={\ }]{e3}.}
\address[D]{Daniela Witten is a Professor in the Departments of Statistics and Biostatistics, University of Washington\printead[presep={\ }]{e4}.}
\end{aug}

\begin{abstract}
    In recent years, there has been substantial interest in the task of selective inference: inference on a parameter that is selected from the data. Many of the existing proposals fall into what we refer to as the \emph{infer-and-widen} framework: they produce symmetric confidence intervals whose midpoints do not account for selection and therefore are biased; thus, the intervals must be wide enough to account for this bias.
In this paper, we investigate infer-and-widen approaches in three vignettes: the winner's curse, maximal contrasts, and inference after the lasso.
In each of these examples, we show that a state-of-the-art infer-and-widen proposal leads to confidence intervals that are wider than a non-infer-and-widen alternative. 
Furthermore, even an ``oracle''  infer-and-widen confidence interval --- the narrowest possible interval that could be theoretically attained via infer-and-widen --- can be wider than the alternative.  
\end{abstract}

\begin{keyword}
\kwd{Selective inference}
\kwd{post-selection inference}
\kwd{data fission}
\kwd{algorithmic stability}
\kwd{infer-and-widen}
\end{keyword}

\end{frontmatter}

\section{Introduction}
\label{sec:intro}

In classical statistics, the target of inference --- for instance, the null hypothesis that we wish to test, or the parameter for which we wish to construct a confidence interval --- is  fixed. However,  motivated by the explosion of data across a variety of fields, recent interest has focused on the setting where the target of inference is itself a function of the data, and therefore is random. 
The terms \emph{selective inference} and \emph{post-selection inference} describe the set of strategies available for  conducting inference on a random target.

To formalize this, consider a random variable $Y $ with domain $\mathcal{Y}$, and a collection of parameters $\theta$ whose elements can be indexed by a (possibly infinite) index set $\Ical$. For a fixed index $i \in \Ical$, and for $\alpha \in (0,1)$, a \emph{classical} $1-\alpha$ confidence interval for $\theta_i$, which we will denote $\CI_i^{\alpha}$, satisfies 
\begin{equation}\label{eq:nominal_coverage}
    \PP \rbr{\theta_i \not \in \CI_i^{\alpha}(Y)} \leq \alpha.
\end{equation}
Now, suppose that  $i \in \mathcal{I}$ is no longer a fixed index; instead, some selection rule $\Scal : \mathcal{Y} \to \Ical$ uses the data to select a parameter on which to conduct inference. A confidence interval provides valid \emph{unconditional coverage}~\citep{berk_valid_2013} at level $\alpha$ if
\begin{equation}\label{eq:unconditional_coverage}
    \PP \rbr{\theta_{\Sel(Y)} \not \in \CI_{\Sel(Y)}^{\alpha}(Y)} \leq \alpha.
\end{equation}
Unconditional coverage is \emph{not} guaranteed by classical confidence intervals. In this paper, we consider methods that yield unconditional guarantees, in the sense of~\cref{eq:unconditional_coverage}. 
Some of these confidence intervals also yield (stronger) \emph{conditional} guarantees~\citep{fithian_optimal_2017}, in the sense that
\begin{equation}\label{eq:conditional_coverage}
    \PP \rbr{\theta_{\Sel(Y)} \not \in \CI_{\Sel(Y)}^{\alpha}(Y)~|~\Ccal(Y)} \leq \alpha
\end{equation}
for an event $\Ccal(Y)$ satisfying $\sigma\rbr{\Ccal(Y)} \supseteq \sigma \rbr{\Scal(Y)}$, where $\sigma(\cdot)$ denotes the sigma algebra generated by a random variable, i.e., the coverage guarantee in~\cref{eq:conditional_coverage} conditions on at least the selection event. By the law of total expectation, conditional coverage, \cref{eq:conditional_coverage}, implies unconditional coverage, \cref{eq:unconditional_coverage}.



    


\subsection{The infer-and-widen framework}

As presented in Box~\ref{box:infer-and-widen}, the infer-and-widen framework involves (Step 1) defining a target of inference using some or all of the data, and then (Step 2) symmetrically adjusting the  critical values arising from ``classical'' inference to account for the fact that the target of inference is not fixed in order to achieve unconditional coverage, \eqref{eq:unconditional_coverage}. The simplest examples of this strategy are the Bonferroni~\citep{bonferroni_teoria_1936, dunn_multiple_1961},  Holm-Bonferroni~\citep{holm_simple_1979}, and \v{S}id\'{a}k \citep{sidak_rectangular_1967} procedures, which are suitable when a finite number of targets of inference are considered (and in the case of \v{S}id\'{a}k, under additional assumptions). In the special case of inference on contrasts of a linear model, the Scheff\'{e} correction is available~\citep{scheffe_analysis_1959}; this idea is further refined by  the PoSI method~\citep{berk_valid_2013}. While conceptually attractive, these approaches tend to be quite conservative (Bonferroni, Holm-Bonferroni, \v{S}id\'{a}k), are limited to a linear model setting (Scheff\'{e}, PoSI), or are often computationally intractable (PoSI).

Other recent infer-and-widen approaches have attempted to address these shortcomings. The \emph{local simultaneous inference} approach of~\citet{zrnic_locally_2024} focuses a simultaneous correction on the subset of hypotheses that are most likely given the observed data. The \emph{algorithmic stability} proposal of \cite{zrnic_post-selection_2023} also adapts to the details of the selection algorithm. By drawing connections with the field of differential privacy \citep{dwork_book_2014}, the ``stability'' of the approach for selecting the target of inference in Step 1 of Box~\ref{box:infer-and-widen} is quantified; this, in turn, leads to a stability-specific adjustment of the critical values in Step 2 of Box~\ref{box:infer-and-widen} (see also Section A of the online supplement). In a sense, these two approaches achieve the lofty goal laid out in the PoSI proposal of \cite{berk_valid_2013}  by creating frameworks to compute selection-specific adjustments to the critical values.

\begin{mybox}[floatplacement=!htb,label={box:infer-and-widen}]{\emph{Infer-and-widen.}} 
\begin{enumerate}
\item A target of inference is selected from the data.
\item ``Classical'' inference is conducted, with a symmetric adjustment of the critical values to account for the data-driven selection of the target of inference. 
\end{enumerate}
\end{mybox}

\subsection{Other approaches for valid selective inference}

Not all approaches that yield unconditional guarantees of the form \eqref{eq:unconditional_coverage} fall into the infer-and-widen framework. Some methods produce \textit{asymmetric} confidence intervals around the biased point estimate: for instance, \citet{benjamini_confidence_2019} and \citet{fuentes_confidence_2018} construct an asymmetric Bonferroni correction that accounts for the direction of the bias. More recently, \citet{zrnic_flexible_2025} extend the local simultaneous approach of \citet{zrnic_locally_2024} by additionally adapting the correction to the specific selection algorithm. The hybrid approaches of~\cite{andrews_inference_2024} and  \cite{mccloskey_hybrid_2024} apportion a fraction of the Type 1 error budget to rule out unlikely selection events, and then construct a conditionally-valid confidence interval on the likely hypotheses; note that this only guarantees unconditional coverage, \cref{eq:unconditional_coverage}. Lastly, there is a line of empirical Bayes work whose intervals are centered at shrinkage estimators which, while likely biased, improve upon the non-selective estimator \citep{hwang_empirical_2013}.

Furthermore, a number of methods yield coverage guarantees conditional on the selection event, i.e., \eqref{eq:conditional_coverage}. Some methods do this by obtaining independent training and test sets; then, the training set can be used to define the target of inference, and the test set to conduct inference on this target. Because the two sets are independent, we can conduct inference using the test set \emph{as though the target of inference were fixed}, despite the fact that it is actually a function of the training data. Sample splitting, applicable  in the setting of independent and identically distributed observations, provides one approach to obtain independent training and test sets \citep{cox_note_1975}. 
More recent  \emph{data thinning} proposals  generalize the idea of sample splitting, and can be applied when only a single observation is available \citep{rasines_splitting_2023, neufeld_data_2024, dharamshi_generalized_2024}. 

In cases where it is either impossible or undesirable to create independent training and test sets, strategies that achieve conditional guarantees of the form~\cref{eq:conditional_coverage} may still be available. For instance, a target of inference can be defined using \emph{all} of the data, and then inference can be conducted \emph{conditional} on the event that the data led to this particular target of inference. Examples include the proposals of~\cite{fithian_optimal_2017, lee_exact_2016, chen_selective_2023, gao_selective_2022, andrews_inference_2024}.
Alternatively, if a researcher is willing to define the target of inference using only a subset of the data or using a noisy version of the data, then \emph{data fission}~\citep{leiner_data_2023} or \emph{randomized conditional selective inference}~\citep{tian_selective_2018, panigrahi_exact_2024} may be applicable. Randomization in this manner preserves information for inference, overcoming a limitation of conditioning on selection using all of the data.


\subsection{Contributions and organization}

In this paper, we conduct an empirical and theoretical evaluation of infer-and-widen strategies in the context of three vignettes. The first two vignettes involve inference after a randomized selection rule, and we derive alternative confidence intervals to compare to an existing infer-and-widen approach. In the third vignette, we compare a suite of existing methods.
The bottom line of our paper is as follows: when comparing infer-and-widen strategies to alternative selective inference methods that make use of \emph{identical} selection events, the latter often outperform the former in terms of confidence interval width. Furthermore, in some settings, the latter even outperform the ``oracle'' infer-and-widen procedure: the narrowest possible interval that could be theoretically attained within the infer-and-widen framework.

The organization is as follows. In Section~\ref{sec:infer_and_widen}, we formalize the infer-and-widen framework. In Sections~\ref{sec:vignette-1},  \ref{sec:vignette-2}, and \ref{sec:vignette-3}, we compare its performance to new alternative intervals that we derive in two vignettes considered in \cite{zrnic_post-selection_2023}, and to existing alternative intervals in an example involving model selection via the lasso. In Sections \ref{sec:vignette-1} and \ref{sec:vignette-2}, the selection event is randomized, whereas in \cref{sec:vignette-3} it is not. In~\cref{sec:vignette-1-bias}, we show analytically that the infer-and-widen interval can suffer from a highly biased midpoint. We close with a  discussion in Section~\ref{sec:discussion}. Theoretical results and additional empirical comparisons are proven in the online supplement. Scripts to fully reproduce all numerical results can be found at \href{https://github.com/rflperry/infer_and_widen}{https://github.com/rflperry/infer\_and\_widen}.




\section{Formalizing the infer-and-widen framework}\label{sec:infer_and_widen}

We now formalize the \emph{infer-and-widen} (I\&W) framework, outlined in Box~\ref{box:infer-and-widen}. This  involves a modification of the classical confidence interval $\CI_{i}^{\alpha}$ in~\eqref{eq:nominal_coverage} to account for the fact that $\theta_{\Sel(Y)}$ is a data-driven parameter. Based on the particulars of the selection function $\Sel(\cdot): \mathcal{Y} \rightarrow \mathcal{I}$ and data $Y$, a modified $\alpha$-level, $H(\alpha, \Scal, Y)$, is derived so that $\CI^{H(\alpha,\Scal,Y)}_{\Sel(Y)}(Y)$ is unconditionally valid at level $\alpha$ in the sense of~\cref{eq:unconditional_coverage}, i.e., 
\begin{equation}\label{eq:infer_and_widen_ci}
    \PP \rbr{\theta_{\Sel(Y)} \not \in \CI_{\Sel(Y)}^{H(\alpha, \Scal, Y)}(Y)} \leq \alpha.
\end{equation}
In effect, this procedure  symmetrically widens the classical confidence interval to account for bias in its midpoint.

In a Bonferroni correction we have $H(\alpha, \Scal,Y) = \alpha / |\mathcal{I}|$, and in the Holm-Bonferroni correction it is slightly larger.  Since these approaches  adapt neither to the selection rule $\Sel(\cdot): \mathcal{Y} \rightarrow \mathcal{I}$ nor to the data $Y$, the constant $H(\alpha, \Scal, Y)$ is easy to compute but  overly conservative. 
Despite the notation $H(\alpha, \Scal,Y)$ used in \eqref{eq:infer_and_widen_ci}, the aforementioned infer-and-widen proposals result in confidence intervals whose widths do not depend on the particulars of the selection rule $\Sel(\cdot)$ nor data $Y$. 

The foundational PoSI proposal~\citep{berk_valid_2013} considers the setting of linear regression with fixed regressors. They derive a constant $H(\alpha, \Scal, Y)$ that enables inference on the coefficient(s) in \emph{any} submodel, regardless of the form of $\Sel(\cdot)$ or observed data $Y$. However, for a particular selection rule $\Sel(\cdot)$, the resulting confidence interval may be quite conservative, and computing the constant $H(\alpha, \Scal,Y)$ is challenging or infeasible.

In contrast are two newer approaches, which provide narrower constants $H(\alpha, \Scal, Y)$ that depend on the selection rule $\Scal(\cdot)$. The locally simultaneous inference method of \citet{zrnic_locally_2024} apportions the Type-1 error budget $\alpha$, first to use the data $Y$ to rule out the most unlikely hypotheses, and then to perform a simultaneous correction over a \textit{subset} of all $\abr{\Ical}$ hypotheses.
The algorithmic stability proposal of \cite{zrnic_post-selection_2023} adjusts the critical value $H(\alpha, \Scal, Y)$ based on how  ``stable'' the selection rule $\Scal(\cdot)$ is; intuitively, a more ``stable'' selection rule requires a less substantial adjustment.
Section A in the online supplement
contains a more detailed overview.


We close with a remark about the infer-and-widen framework:
at first glance, it appears that an  ``advantage'' of infer-and-widen, over alternative selective inference approaches, is that the interval in \eqref{eq:infer_and_widen_ci} is of the same form as the classical interval in the absence of selection \eqref{eq:nominal_coverage}, i.e., no bespoke methods for computing a confidence interval are required.  However, this is a double-edged sword, in the sense that \emph{bias induced by the selection event is perpetuated into the confidence interval}.
 Thus, an infer-and-widen  confidence interval that attains the nominal coverage may be \emph{extremely} wide, to account for bias in the interval's midpoint.
We will see that this is a major drawback of  infer-and-widen. 

\section{Vignette \#1: The winner's curse}
\label{sec:vignette-1}

Suppose that we observe $n$ independent Gaussian random variables $Y_1,\ldots,Y_n$ with common known variance, and are interested in conducting inference on the unknown mean of the largest observed variable. That is, 
\begin{equation}\label{eq:vignette-1}
    Y \sim \Norm_n(\mu, \sigma^2 I_n), 
\end{equation}
and the selection rule is 
\begin{equation}
 \Sel(Y) := \argmax_{i \in [n]} Y_i.
 \label{eq:vignette-1-selection}
\end{equation}
It is not hard to see that a $1-\alpha$ confidence interval for $\mu_{\Sel(Y)}$ that
treats $\Sel(Y)$ as though it were fixed (rather than a function of the data) will fail to achieve the nominal coverage. This is the ``winner's curse''~\citep{andrews_inference_2024}.

To achieve a ``stable'' algorithm, \cite{zrnic_post-selection_2023}  consider a randomized version of the  selection rule \eqref{eq:vignette-1-selection} via the addition of centered Laplace noise,
\begin{equation}
 \SelL(Y) := \argmax_{i \in [n]} (Y_i +  \zeta_i),
\label{eq:vignette-1-selection-laplace}
\end{equation}
where $\zeta_1,\ldots,\zeta_n \indsim \mathrm{Laplace}\left(c\right)$ and $c$ is a positive constant that governs the amount of additive noise. When $c=0$, $\SelL(Y)$ simplifies to $\Sel(Y)$ in \eqref{eq:vignette-1-selection}. 

\begin{proposition}
[An infer-and-widen interval based on \eqref{eq:vignette-1-selection-laplace}, given by \citep{zrnic_post-selection_2023}]
\label{thm:vignette-1-laplace-AS} Under~\cref{eq:vignette-1}, consider the selection rule $\SelL(\cdot)$ defined in \eqref{eq:vignette-1-selection-laplace}. For any $\eta > 0$ and $\nu \in (0, \alpha)$ such that $ \eta c \geq 2 z_{1 - \alpha (\alpha - \nu)/2n }$, where $z_q$ denotes the $q$th quantile of the standard normal distribution, the infer-and-widen confidence interval
     \begin{equation} 
        \label{eq:vignette-1-CI-LAS}
\CI_{{{\SelL}(Y)}}^{\alpha(1-\alpha + \nu)e^{-\eta}}(Y) := \left[  Y_{{\SelL}(Y)} \pm \sigma z_{1-\alpha(1-\alpha + \nu)e^{-\eta}/2}  \right]
    \end{equation}
    has $1-\alpha$ unconditional coverage for $\mu_{\SelL(Y)}$ in the sense of~\cref{eq:infer_and_widen_ci}.
\end{proposition}

In the context of \cref{thm:vignette-1-laplace-AS}, the selection algorithm in \cref{eq:vignette-1-selection-laplace} is ``stable'' at level $\eta$ with probability at least $1-\nu$. We defer a more detailed discussion to Section A in the online supplement.

We will now derive an alternative interval that makes use of \emph{exactly} the same selection event, \eqref{eq:vignette-1-selection-laplace}, using ideas from the data fission proposal of \cite{leiner_data_2023}. 
Towards this end, note that $\SelL(Y) = \Sel(\Ytrain)$ where $\Ytrain := Y + \zeta$.  
Since $Y\mid \Ytrain$ does not admit a tractable distribution, we will additionally condition on $\rbr{\sign(\zeta_1),\ldots,\sign(\zeta_n)}$. This yields a multivariate truncated normal distribution with independent marginals. 

\begin{proposition}[A data fission interval based on \eqref{eq:vignette-1-selection-laplace}]\label{prop:vignette-1-laplace-DF}
    Under \eqref{eq:vignette-1}, consider the selection rule $\SelL(Y)$ defined in \eqref{eq:vignette-1-selection-laplace}. Define $\delta_i := \sign(\zeta_i)$ and $A_i := \{y_i \in \RR: \delta_i y_i \leq \delta_i (Y_i + \zeta_i) \}$. Let $\Phi_{\mu, \sigma^2, A_i}(\cdot)$ denote the CDF of a $\Norm(\mu, \sigma^2)$ distribution truncated to the set $A_i$, and define $a_i$ and $b_i$ to satisfy
    \begin{equation}
        \Phi_{a_i + \delta_{i} \sigma^2 / c, \sigma^2, A_i}(Y) = 1 - \alpha / 2
        \label{eq:vignette-1-a}
    \end{equation}
    and
    \begin{equation}
        \Phi_{b_i + \delta_{i} \sigma^2 / c, \sigma^2, A_i}(Y) = \alpha/2.
        \label{eq:vignette-1-b}
    \end{equation}
    The confidence interval 
     \begin{equation} 
        \label{eq:vignette-1-CI-LDF}
       \CI_{\SelL(Y)}^{\alpha}(Y)
       :=   \sbr{ a_{\SelL(Y)}, b_{\SelL(Y)} }
    \end{equation}
    has $1-\alpha$ conditional coverage for $\mu_{\SelL(Y)}$ in the sense of~\cref{eq:conditional_coverage} with $\Ccal(Y) = \rbr{Y+\zeta, \delta}.$
\end{proposition}
Computing \eqref{eq:vignette-1-CI-LDF} requires numerically solving \eqref{eq:vignette-1-a} for $a_{i}$ and \eqref{eq:vignette-1-b} for $b_{i}$, where $i = \SelL(Y)$. 

Thus, both \eqref{eq:vignette-1-CI-LAS} and \eqref{eq:vignette-1-CI-LDF} provide confidence intervals for $\mu_{\SelL(Y)}$ that attain $1-\alpha$ unconditional coverage. We should prefer the approach that yields narrower intervals. 
Figure~\ref{fig:vignette-1}(a) displays the ratio of the widths of the two intervals when $\mu = 0$ and $\sigma^2 = 1$, as we vary the sample size $n$ and the variance $2c^2$ of the Laplace noise in \eqref{eq:vignette-1-selection-laplace}.
 
The infer-and-widen interval \eqref{eq:vignette-1-CI-LAS} involves tuning parameters $\nu$ and $\eta$; to construct Figure~\ref{fig:vignette-1}(a), for a fixed value of the noise $c$ and sample size $n$ in \eqref{eq:vignette-1-selection-laplace}, we select the $\nu$ and $\eta$ that minimize the confidence interval width, subject to the constraints in Proposition~\ref{thm:vignette-1-laplace-AS}. The data fission intervals \eqref{eq:vignette-1-CI-LDF} are much narrower than the infer-and-widen intervals \eqref{eq:vignette-1-CI-LAS}, across all values of $n$ and $c$ considered. 

It is possible that the width of the infer-and-widen interval \eqref{eq:vignette-1-CI-LAS} is due to the fact that \cite{zrnic_post-selection_2023} apply a series of inequalities in its derivation; furthermore, it is natural to wonder whether another infer-and-widen proposal could perform better. Therefore, we now consider the ``oracle'' infer-and-widen interval:  this is the narrowest possible interval centered at $Y_{\SelL(Y)}$ that achieves a given coverage. It is not computable in practice, as it requires knowledge of the true parameter value $\mu$.
In Figures~\ref{fig:vignette-1}(b, c), we numerically compare the ``oracle'' infer-and-widen method to the data fission interval \eqref{eq:vignette-1-CI-LDF} under the  selection rule \eqref{eq:vignette-1-selection-laplace} with variances $2c^2=0.33$ and $3$, respectively. The former method has a wider average width than the latter: thus, \emph{no infer-and-widen interval can outperform the data fission interval \eqref{eq:vignette-1-CI-LAS}} in this setting.

Section C of the supplementary materials contains a related comparison of methods for inference after non-randomized selection using \cref{eq:vignette-1-selection}. The conclusions are the same.

\begin{figure}[!htb]
    \centering
    \includegraphics[width=\linewidth]{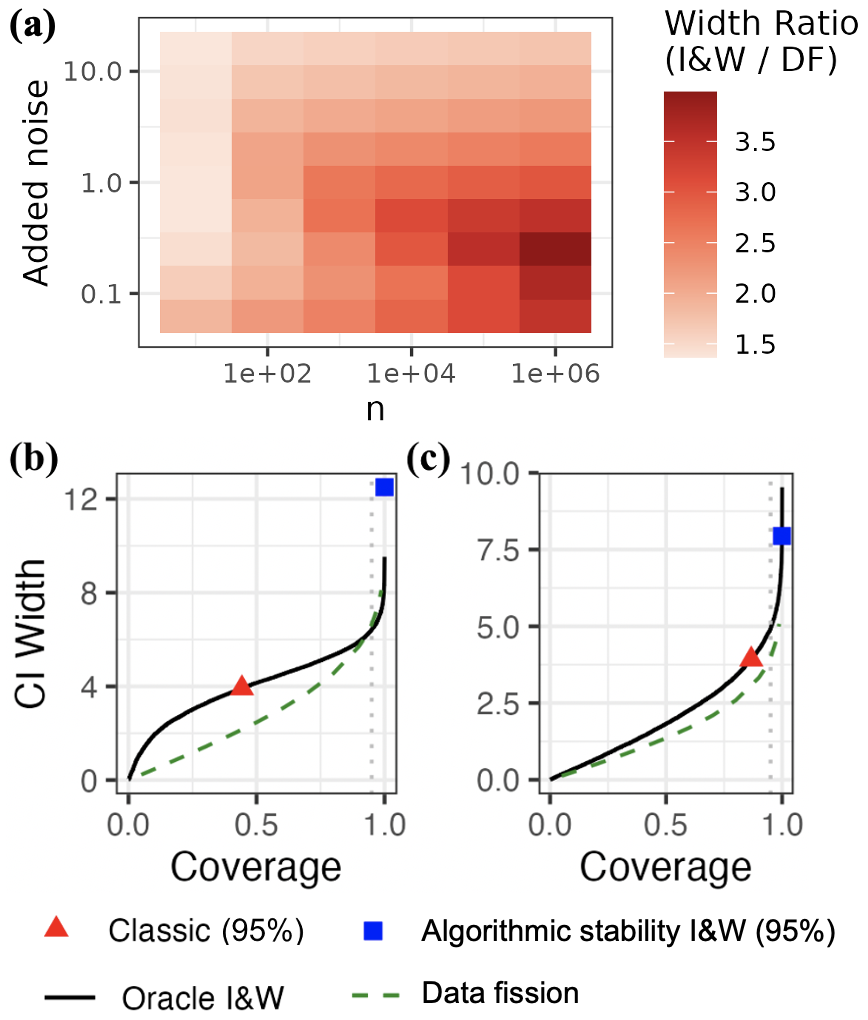}  
    \caption{Vignette \#1 under \eqref{eq:vignette-1} with $\mu=0$, $\sigma=1$, and the selection rule in~\cref{eq:vignette-1-selection-laplace}. Results are averaged over (a) $250$ or (b, c) $1000$ simulated datasets. \textbf{(a):} We display the ratio of the average widths of the infer-and-widen (I\textup{\symbol{`\&}}W) \cref{eq:vignette-1-CI-LAS} and data fission (DF)~\cref{eq:vignette-1-CI-LDF} intervals with $95\%$ coverage across values of the sample size $n$ and Laplacian noise variance $2c^2$. A width ratio that exceeds one implies that the infer-and-widen interval is wider. \textbf{(b, c):} Fixing $n=100$, we display the average width of the ``oracle'' infer-and-widen interval and the data fission interval~\cref{eq:vignette-1-CI-LDF} as a function of coverage. The 95\% classical interval and the ``attainable'' 95\% infer-and-widen interval based on algorithmic stability~\cref{eq:vignette-1-CI-LAS} are also displayed. The Laplacian noise variance is low and high in panels (b) and (c), respectively.
    }
    \label{fig:vignette-1}
\end{figure}

\section{Vignette \#2: Maximal contrasts}\label{sec:vignette-2}

Following~\citet{zrnic_post-selection_2023}, we consider the model~\cref{eq:vignette-1} and suppose we also have access to a fixed design matrix with normalized columns $\nbr{X_1}_2 = \dots = \nbr{X_p}_2 = 1$.
Our interest lies in the quantity $X_{\Seltwo(Y)}^\top \mu$ defined via the selection rule
\begin{equation}\label{eq:vignette-2-selection}
    \Seltwo(Y) := \argmax_{j \in [p]} \abr{X_j^\top Y};
\end{equation}
\cite{zrnic_post-selection_2023} refer to this as the ``maximal contrast''. If $X=I$, then this setting is similar to Vignette \#1, except with an absolute value in the selection rule.

To achieve a ``stable'' algorithm, \cite{zrnic_post-selection_2023} consider the selection rule
\begin{equation} \label{eq:vignette-2-selection-laplace}
    \SelLtwo(Y) := \argmax_{j \in [p]} \abr{X^\top_j Y +  \zeta_j},
\end{equation}
where $\zeta_1,\ldots,\zeta_p \iidsim \mathrm{Laplace}\left(c\right)$ and $c$ is a positive constant that governs the scale of the additive noise.

\begin{proposition}[An infer-and-widen interval based on \eqref{eq:vignette-2-selection-laplace}, given by~\citet{zrnic_post-selection_2023}]\label{thm:vignette-2-laplace-AS}
    Under \eqref{eq:vignette-1},  consider the Laplacian selection rule $\SelLtwo(Y)$ defined in \eqref{eq:vignette-2-selection-laplace}. For any $\eta > 0$ and $\nu \in (0, 1)$ such that $c \geq 2 z_{1 - \alpha\nu/(2p) }/\eta$, the infer-and-widen confidence interval
    \begin{equation} \label{eq:vignette-2-CI-LAS}
        \CI_{{{\SelLtwo}(Y)}}^{\alpha(1-\nu)e^{-\eta}}(Y) := \left[  X^\top_{\SelLtwo(Y)} Y \pm \sigma z_{1-\alpha(1-\nu)e^{-\eta}/2}  \right]
    \end{equation}
    has $1-\alpha$ unconditional coverage for $X^\top_{\SelLtwo(Y)}\mu$ in the sense of~\cref{eq:infer_and_widen_ci}.
\end{proposition}
We will take a randomized conditional inference approach~\citep{tian_selective_2018, panigrahi_exact_2024} to derive an alternative confidence interval that makes use of the identical selection event.

\begin{proposition}[A randomized conditional selective inference interval based on~\cref{eq:vignette-2-selection-laplace}]\label{prop:vignette-2-CI-LDF}
    Consider the model in~\cref{eq:vignette-1} and selection rule $\SelLtwo(Y)$ in~\cref{eq:vignette-2-selection-laplace}. Define
    \begin{align*}
        W &:= (I_n - X_{\SelLtwo(Y)} X_{\SelLtwo(Y)}^\top) Y, \\
        \delta &:= \sign\rbr{X^\top_{\SelLtwo(Y)} Y + \zeta_{\SelLtwo(Y)}},
    \end{align*}
    and $a_j$ and $b_j$ to satisfy
    \begin{equation}
        \Phi_{a_j, \sigma^2, [v_{\min}, v_{\max}]}(X_j^\top Y) = 1 - \alpha / 2
        \label{eq:vignette-2-a}
    \end{equation}
    and
    \begin{equation}
        \Phi_{b_j, \sigma^2, [v_{\min}, v_{\max}]}(X_j^\top Y) = \alpha / 2,
        \label{eq:vignette-2-b}
    \end{equation}
    where $\Phi_{\lambda, \sigma^2, [v_{\min}, v_{\max}]}(\cdot)$ is defined in~\cref{prop:vignette-1-laplace-DF}, and $v_{\min}$ and $v_{\max}$ are defined in~\cref{sec:deferred}. The confidence interval
     \begin{equation} 
        \label{eq:vignette-2-CI-LDF}
       \CI_{\SelLtwo(Y)}^{\alpha} \left(Y \right)
       := \sbr{ a_{\SelLtwo(Y)}, b_{\SelLtwo(Y)} }
    \end{equation}
    has $1-\alpha$ conditional coverage for $X_{\SelL(Y)}^\top \mu$ in the sense of~\cref{eq:conditional_coverage} with $\Ccal(Y) = \rbr{\SelLtwo(Y), \zeta, W, \delta}$.
\end{proposition}
Computing \cref{eq:vignette-2-CI-LDF} requires numerically solving \cref{eq:vignette-2-a} for $a_{i}$ and \cref{eq:vignette-2-b} for $b_{i}$, where $i = \SelLtwo(Y)$. 

We simulate data according to \cref{eq:vignette-1} with $n=100$ and $\sigma=1$. We set $\mu = X \phi$ where half of the elements of $\phi$ are zero and the other half are drawn from an exponential distribution with mean $50/7$. The rows of $X$ are independent draws from $\Norm_p(0, \rho (1_p1_p^\top  - I_p) + I_p)$ with $\rho=0.5$, and its columns are subsequently scaled so that $\nbr{X_1}_2 = \dots = \nbr{X_p}_2 = 1$. \cref{fig:vignette-2}(a) displays the ratio of confidence interval widths across values of $p$ and the variance of the Laplacian noise $2c^2$ in \cref{eq:vignette-2-selection-laplace}.
The infer-and-widen interval \eqref{eq:vignette-2-CI-LDF} involves tuning parameters $\nu$ and $\eta$; to construct Figure~\ref{fig:vignette-2}(a), for a fixed value of the noise $c$ and number of features $p$ in \eqref{eq:vignette-2-selection-laplace}, we select the $\nu$ and $\eta$ that minimize the confidence interval width, subject to the constraints in Proposition~\ref{thm:vignette-2-laplace-AS}. The randomized conditional selective inference intervals are narrower than the infer-and-widen intervals across most values of $p$ and $c$ considered.

Once again, we consider the ``oracle'' infer-and-widen interval:  this is the narrowest possible interval centered at $X_{\SelLtwo(Y)}^\top Y$ that achieves a given coverage.
In Figures~\ref{fig:vignette-2}(b, c), we numerically compare the ``oracle'' infer-and-widen method to the randomized conditional selective inference interval \eqref{eq:vignette-2-CI-LDF} under the selection rule \eqref{eq:vignette-2-selection-laplace} with variances $2c^2=0.33$ and $3$, respectively. In this setting, the ``oracle'' infer-and-widen interval is narrower than the randomized conditional selective inference interval in \eqref{eq:vignette-2-CI-LDF}, thereby indicating that in principle, an infer-and-widen interval can outperform the randomized conditional selective inference interval \eqref{eq:vignette-2-CI-LDF} in this setting. However, to the best of our knowledge, no methods to construct such a narrower infer-and-widen interval are currently available.

\begin{figure}[!htb]
    \centering
    \includegraphics[width=\linewidth]{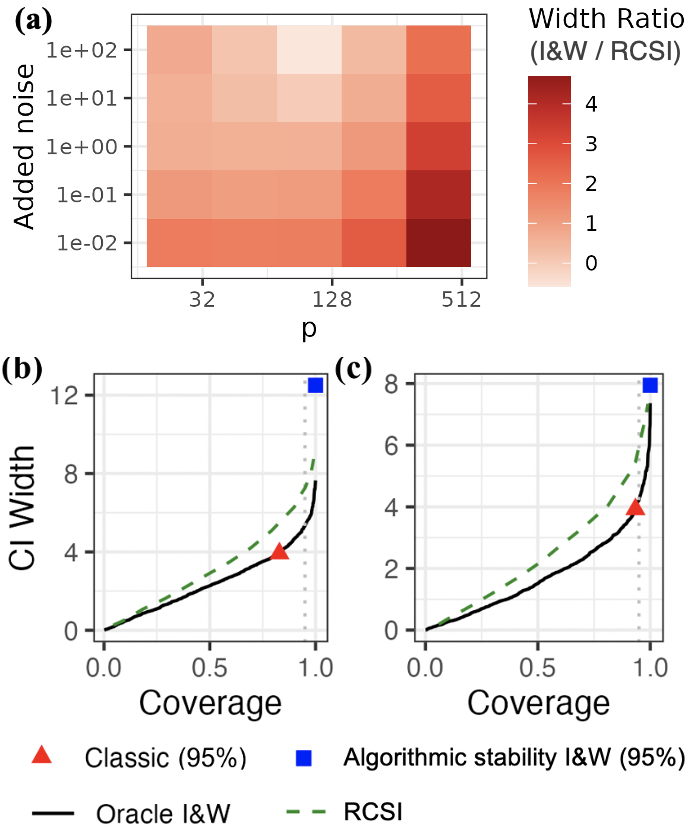}  
    \caption{Vignette \#2 under \eqref{eq:vignette-1} with $n=100$, $\sigma=1$, $\mu=X \phi$, and the selection rule in~\cref{eq:vignette-2-selection-laplace}, where $X$ is multivariate normal with column correlation $0.5$, and $\phi$ is sparse with exponentially distributed non-zero elements. Results are averaged over (a) $250$ or (b, c) $1000$ simulated datasets. \textbf{(a):} We display the ratio of the average widths of infer-and-widen (I\textup{\symbol{`\&}}W)~\cref{eq:vignette-2-CI-LAS} and randomized conditional selective inference (RCSI)~\cref{eq:vignette-2-CI-LDF} intervals with $95\%$ coverage across the number of features $p$ and Laplacian noise variance $2c^2$. \textbf{(b, c):} Fixing $p=100$, we display the average width of the ``oracle'' infer-and-widen interval and the RCSI interval~\cref{eq:vignette-2-CI-LDF} as a function of coverage. The 95\% classical interval and the ``attainable'' 95\% infer-and-widen interval~\cref{eq:vignette-2-CI-LAS} based on algorithmic stability are also displayed. The Laplacian noise variance is low and high in panels (b) and (c), respectively.
    }
    \label{fig:vignette-2}
\end{figure}

\section{Vignette \#3: Inference after the lasso}\label{sec:vignette-3}

We now consider inference on the set of coefficients selected by the lasso in a linear model. We again consider the model~\cref{eq:vignette-1}, alongside a fixed design matrix $X \in \RR^{n \times p}$. Our selection rule is the set of features selected by the lasso estimator $\hat\beta^{(\lambda)}$~\citep{tibshirani1996regression}:
\begin{equation}
\label{eq:vignette-3-selection}
\Selthree(Y) := \supp\rbr{\hat\beta^{(\lambda)}},
\end{equation}
where
\begin{equation*}
    \hat\beta^{(\lambda)} := \argmin_\beta \left\{ \tfrac1n \nbr{ Y - X \beta }^2_2 + \lambda \|  \beta \|_1 \right\}\nonumber.
\end{equation*}
This is a ``non-randomized'' selection event, and we defer to Section E of the supplement for inference after a related randomized selection event.

Our goal is valid inference on the least squares coefficients in the model selected by $\Selthree(Y)$:
\begin{equation*}
    \beta_{\Selthree(Y)} :=  \rbr{X_{\Selthree(Y)}^\top X_{\Selthree(Y)}}^{-1} X_{\Selthree(Y)}^\top \EE\sbr{Y},
\end{equation*}
where $X_{\Selthree(Y)}$ is the subset of columns of $X$ indexed by $\Selthree(Y)$. Since $\Selthree(Y)$ is a data-adaptive model, the naive point estimator
\begin{equation*}
    \hat \beta_{\Selthree(Y)} :=  \rbr{X_{\Selthree(Y)}^\top X_{\Selthree(Y)}}^{-1} X_{\Selthree(Y)}^\top Y
\end{equation*}
is no longer guaranteed to be unbiased for $\beta_{\Selthree(Y)}$.

Obtaining valid inference on $\beta_{\Selthree(Y)}$ in the sense of \cref{eq:unconditional_coverage} is a classic problem in the field of selective inference, and there are two main categories of methods which do so. The first category guarantees unconditional coverage, \cref{eq:unconditional_coverage}, via the stronger goal of simultaneous coverage. Since the Bonferroni correction~\citep{dunn_multiple_1961} is valid but would require an overly conservative adjustment for all $O(2^p)$ possible models, \citet{berk_valid_2013} derive the narrower PoSI simultaneous correction. \citet{zrnic_locally_2024} subsequently developed the locally simultaneous inference approach which focuses the simultaneous correction on only a subset of likely models, thereby often improving power; these two approaches yield infer-and-widen intervals centered at $\hat \beta_{\Selthree(Y)}$ and with widths proportional to the diagonal elements of $\rbr{X_{\Selthree(Y)}^\top X_{\Selthree(Y)}}^{-1/2}$ with a shared constant of proportionality.
The second category of methods approach unconditional coverage via the stronger goal of conditional coverage, \cref{eq:conditional_coverage}. The foundational work of \citet{lee_exact_2016} characterized the selection event $\Selthree(Y)$, and used this to derive the conditional distribution $\hat \beta_{\Selthree(Y)} \mid \Ccal(Y)$ for inference on $\beta_{\Selthree(Y)}$, where the sigma algebra of $\Ccal(Y)$ is a superset of that of $\Scal(Y)$. In order to improve upon the power of this method, \citet{mccloskey_hybrid_2024} developed the hybrid method, which uses the conditional confidence interval in cases where there is strong signal, and otherwise resorts to unconditional intervals; thus the hybrid method provides only unconditional and not conditional coverage.

As before, we compare the widths of these methods' confidence intervals, as well as the ``oracle'' infer-and-widen interval: the narrowest possible interval centered at $\hat \beta_{\Selthree(Y)}$ and with width proportional to the diagonal elements of $\rbr{X_{\Selthree(Y)}^\top X_{\Selthree(Y)}}^{-1/2}$ with a shared constant of proportionality that attains coverage at a given level. We simulate data according to \cref{eq:vignette-1} with $n=100$, $p=10$, $\sigma=1$, and fixed penalty $\lambda$. We set $\mu = X \beta$, where $\beta_1 = \dots = \beta_5$ and $\beta_6 = \dots = \beta_{10} = 0$. The rows of $X$ are independent draws from $\Norm_p(0, \rho (1_p1_p^\top  - I_p) + I_p)$ with correlation $\rho$.
\cref{fig:vignette-3}(a) displays the ratio of average $95\%$ confidence interval widths of the locally simultaneous inference (LSI)~\citep{zrnic_locally_2024} and hybrid~\citep{mccloskey_hybrid_2024} methods across values of $\rho$ and $\beta_1 / \lambda$. A ratio greater than one indicates that the LSI method has a wider average interval width. In \cref{fig:vignette-3}(b, c), we fix $\rho = 0.9$, and set $\beta_1 /\lambda = 0$ in panel (b) and $\beta_1 / \lambda = 1$ in panel (b). We numerically compare the ``oracle'' interval to the conditional~\citep{lee_exact_2016} and hybrid~\citep{mccloskey_hybrid_2024} intervals. For reference, we also display the average empirical widths and coverages of the $95\%$ intervals from the classic, simultaneous PoSI (SI)~\citep{berk_valid_2013}, and LSI~\citep{zrnic_locally_2024} approaches. In this setting, the ``oracle'' infer-and-widen interval is the narrowest interval, thereby indicating that in principle, an infer-and-widen interval can outperform the hybrid interval. However, the hybrid method outperforms both existing infer-and-widen methods.

\begin{figure}[!htb]
    \centering
    \includegraphics[width=\linewidth]{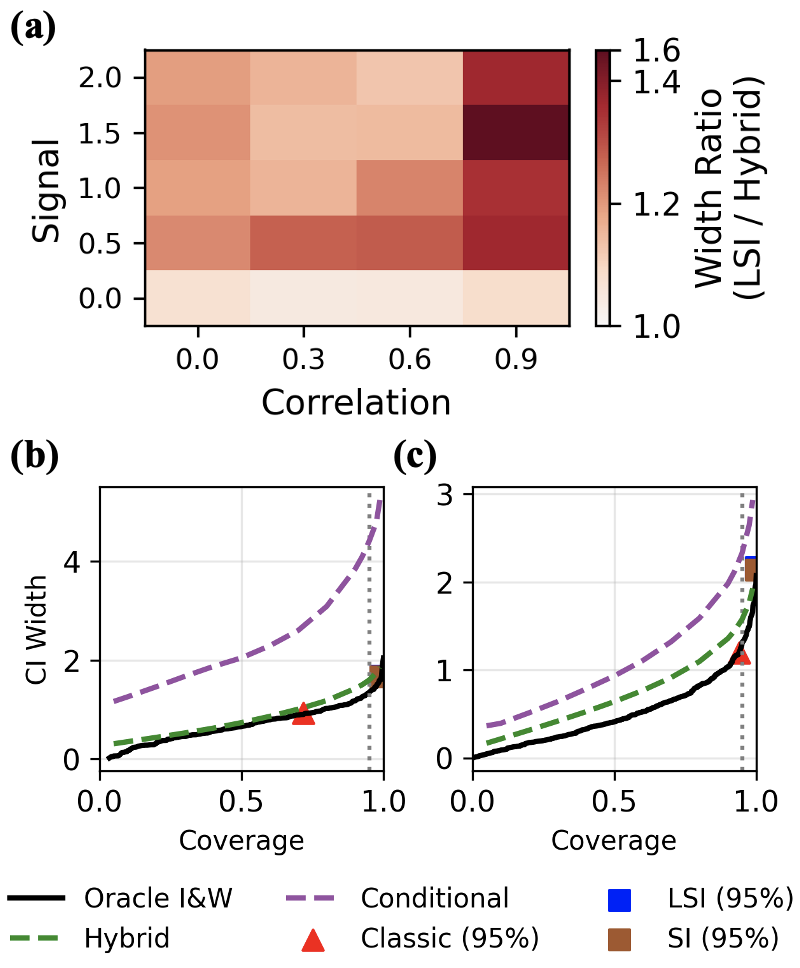}  
    \caption{Vignette \#3 under \eqref{eq:vignette-1} with $n=100$, $p=10$, $\sigma=1$, $\mu=X \beta$, and the selection rule in~\cref{eq:vignette-3-selection}, where the rows of $X$ are equi-correlated normally distributed, and $\beta$ has five nonzero elements. Results are averaged over $100$ simulated datasets. \textbf{(a):} We display the ratio of the average widths of $95\%$ confidence intervals from the LSI infer-and-widen method~\citep{zrnic_locally_2024} and the hybrid method~\citep{mccloskey_hybrid_2024}. A width ratio that exceeds one implies that the infer-and-widen interval is wider. \textbf{(b, c):} Fixing $\rho = 0.9$, we compare the average width of the ``oracle'' infer-and-widen (I\textup{\symbol{`\&}}W) interval to the conditional~\citep{lee_exact_2016} and hybrid~\citep{mccloskey_hybrid_2024} intervals. For reference, we also display the average empirical widths and coverages of the $95\%$ intervals from the classic, simultaneous PoSI (SI)~\citep{berk_valid_2013}, and LSI~\citep{zrnic_locally_2024} approaches. Panels (b) and (c) contain zero and nonzero signal, respectively.
    }
    \label{fig:vignette-3}
\end{figure}

\section{The bias of infer-and-widen intervals}\label{sec:vignette-1-bias}

We now investigate why infer-and-widen intervals are too wide. 
To facilitate the analysis, we consider Vignette \#1 with a  simple selection rule that uses Gaussian noise, 
\begin{equation}
 \SelG(Y) := \argmax_{i \in [n]} (Y_i + c \zeta_i),  \;\;\;\; \zeta_i \iidsim \Norm(0, \sigma^2).
\label{eq:vignette-1-selection-gaussian}
\end{equation}

We know from~\cref{sec:infer_and_widen} that any infer-and-widen interval is centered around $Y_{\SelG(Y)}$.
 \begin{proposition}[Bounding the bias of the infer-and-widen midpoint]\label{prop:bias_bound_gaussian_winner}
    Under \eqref{eq:vignette-1}, with $\mu_1 = \dots = \mu_n$, consider the selection rule $\SelG(Y)$  in \eqref{eq:vignette-1-selection-gaussian}. Then
    \begin{equation*}
        \EE\left[ Y_{\SelG(Y)} - \mu_{\SelG(Y)} \right] \geq \frac{\sigma}{2} (\sqrt{1+c^2} - c) \sqrt{\log n}.
    \end{equation*}
\end{proposition}
That is, the bias increases on the order of $\sqrt{\log n}$. Note that as $c \to \infty$, i.e., as the selection rule becomes fully random, this lower bound goes to zero as expected.

Noting that $\SelG(Y)$ in \eqref{eq:vignette-1-selection-gaussian} is independent of $Y-\frac{1}{c} \zeta$ \citep{neufeld_data_2024}, it is straightforward to show that $[(Y-\frac{1}{c} \zeta)_{\SelG(Y)} \pm \sigma (1 + c^{-2})^{1/2} z_{1-\alpha/2}]$ is a valid $1-\alpha$ conditional confidence interval for $\mu_{\SelG(Y)}$ in the sense of~\cref{eq:conditional_coverage}. Furthermore, $$
    \EE\left[\left(Y-\frac{1}{c} \zeta \right)_{\SelG(Y)}  - \mu_{\SelG(Y)}  \right] = 0,$$
i.e., the midpoint is unbiased for the mean of the selected observation. 

Thus, infer-and-widen here leads to a confidence interval that is centered in the wrong place, but a simple alternative that makes use of the same selection event yields an interval that is correctly centered.

\section{Discussion}
\label{sec:discussion}

In this paper, we compared the performance of the ``infer-and-widen'' approach to simple alternative methods in three vignettes: two taken from \cite{zrnic_post-selection_2023}, and the third involving feature selection via the lasso. Our findings are as follows:
\begin{enumerate}
    \item In Vignette \#1, the ``oracle'' interval was wider than the alternative approaches we consider; thus, there may be limited value in developing narrower infer-and-widen methods that can be computed in practice.
    \item In Vignettes \#2 and \#3, the ``oracle'' interval is the narrowest interval, but infer-and-widen intervals that can be computed in practice are often much wider.
\end{enumerate}



One may wonder whether our findings conflict with those  of \citet{goeman_selection_2024}, who proved that any method that guarantees conditional coverage, \eqref{eq:conditional_coverage}, is dominated by an unconditional multiple testing procedure defined on the full universe of hypotheses. However, their setting is quite different from ours: they consider inference on non-selected parameters and unconditional procedures outside of the infer-and-widen framework.




\begin{appendix}

\section{Deferred notation}\label{sec:deferred}
In the context of~\cref{prop:vignette-2-CI-LDF} , we define
\begin{align}
    A := \begin{bmatrix}
        0_{1 \times n_1} & -1 & 0_{1 \times n_2} \\
        I_{n_1} & -1_{n_1} & 0_{n_1 \times n_2} \\
        0_{n_2 \times n_1} & -1_{n_2} & I_{n_2} \\
        -I_{n_1} & -1_{n_1} & 0_{n_1 \times n_2} \\
        0_{n_2 \times n_1} & -1_{n_2} & -I_{n_2} 
    \end{bmatrix}, \nonumber
\end{align}
where we denote $n_1 := \SelLtwo(y; \zeta)-1$ and $n_2 := p - \SelLtwo(y; \zeta)$ for brevity. Then we can define
\begin{align*}
    v_j &:= -(A\zeta + A X^\top W)_j / (A X^\top X_{\SelLtwo(Y)})_j,\\
    v_{\max} &:= \underset{j : \rbr{\delta A X^\top X_{\SelLtwo(Y)}}_j < 0}{\mathrm{max}} v_j,\\
    v_{\min} &:= \min_{j : \rbr{\delta A X^\top X_{\SelLtwo(Y)}}_j > 0} v_j.
\end{align*}

\end{appendix}
%
%


\begin{funding}

R.P. was partially supported by an Amazon AI research award. D.W. acknowledges funding from NSF DMS 2322920, NSF DMS 2514344, NIH 5P30DA048736, ONR N00014-23-1-2589, and a Simons Investigator Award for Mathematical Modeling of Living Systems.
\end{funding}

\begin{supplement}
\stitle{Supplement to ``Infer-and-widen, or not?''}
\sdescription{Contains proofs of main results, and additional simulation results and figures.}
\end{supplement}


\bibliographystyle{imsart-number} 
\bibliography{references}       


\end{document}


\maketitle

\doublespacing
\appendix

\section{A summary of algorithmic stability}\label{algstab:review}
The notion of \emph{stability} underlying the algorithmic stability proposal of~\citet{zrnic_post-selection_2023} is as follows. 

\begin{definition}[Definitions 4.1 and 4.2, \citet{zrnic_post-selection_2023}]\label{def:stability}
    Two random variables $Q$ and $T$ are 
    \emph{$(\eta, \tau)$-indistinguishable} from each other, denoted by $Q \approx_{(\eta, \tau)} T$, if for any measurable set $\Ocal$, it holds that
    \begin{equation}
        \PP (Q \in \Ocal) \leq e^{\eta} \PP (T \in \Ocal) + \tau.
    \end{equation} 
   A (possibly randomized) selection rule $\Scal : \Ycal \mapsto \Ical$  is \emph{$(\eta, \tau, \nu)$-stable} with respect to a distribution $\Pcal$ on $\mathcal{Y}$ if there exists a (possibly random, possibly dependent on $\Pcal$) variable ${S}_0$ such that
    \begin{equation} \label{eq:stable}
        \PP\rbr{ \Sel(Y) \approx_{(\eta, \tau)} {S}_0} \geq 1 - \nu.
    \end{equation}
    Note that the probability is taken solely with respect to $Y \sim \mathcal{P}$.
\end{definition}
Since $S_0$ is independent of $Y$, we can use the classical interval (1.1) 
to conduct inference on a  target given by $S_0$, i.e., it holds that 
 $\PP \rbr{\theta_{S_0} \not \in \CI_{S_0}^{\alpha}(Y)} \leq \alpha$.
The key insight underlying algorithmic stability is that if $\Sel(Y)$ is stable, then for most $Y \sim \mathcal{P}$ the value of $\Sel(Y)$ is close (in the sense of~\cref{def:stability}) to $S_0$. Thus, it is ``almost as if'' the selection rule $\Sel(Y)$ did not involve $Y$. This notion is formalized in~\cref{thm:alg_stability}.

\begin{theorem}[Theorem 5.2, \cite{zrnic_post-selection_2023}]\label{thm:alg_stability}
    Let $\Sel(\cdot)$ be a $(\eta, \tau, \nu)$-stable selection rule with parameters possibly dependent on the data distribution $\Pcal$ of $Y$. 
    Assume that for any fixed $\theta_i \in \Theta$, we have access to a confidence interval $\CI_{i}^{\alpha}(Y)$ for $\theta_i$ with coverage at least $1 - \alpha$, i.e., it satisfies (1.1). 
    Then 
    \begin{equation}
        \PP \rbr{\theta_{\Sel(Y)} \not\in \CI_{\Sel(Y)}^{(\alpha - \tau - \nu)e^{-\eta}}(Y)} \leq \alpha.
        \label{eq:AS-thm}
    \end{equation}
\end{theorem}
Thus,  to conduct inference on a data-driven target $\theta_{\Sel(Y)}$, we can apply the infer-and-widen interval in  (2.1) 
with $H(\alpha, \Scal)=(\alpha - \tau - \nu)e^{-\eta}$,  a function of the stability of the  selection rule $\Sel(\cdot)$.  
The values of $\eta, \tau,$ and $\nu$ that govern the stability of $\Sel(\cdot)$ directly determine the width of the confidence interval: the smaller these parameters, the narrower the confidence interval. \cite{zrnic_post-selection_2023} construct stable selection rules in a number of settings via the introduction of auxiliary randomness. 

\section{Vignette \#1}

\subsection{Proof of Proposition 2}

To construct a confidence interval with coverage as in (1.3)  
under the selection rule (3.2), 
it suffices to characterize the distribution $Y \mid \Ytrain$, where $\Ytrain := Y + \zeta$ and $\zeta_1,\ldots,\zeta_n \indsim \mathrm{Laplace}\left(c\right)$. In the language of \cite{leiner_data_2023}, this is the ``data fission'' approach. This conditional distribution is difficult to characterize, however, and so we instead consider the distribution conditional on $\Ytrain$ and an additional event.

\begin{lemma}[Data fission with Laplace noise]\label{prop:conditional_distribution_MVN}
    Let $Y \sim \Norm_n(\mu, \sigma^2 I_n)$ where $\mu \in \RR^n$.
    Let $\zeta \in \RR^n$ where $\zeta_i \iidsim \Lap(c)$. Define $\Ytrain := Y + \zeta$, $\delta := \sign(\Ytrain-Y)$, and the set $A := \cbr{y \in \RR^n: \delta_i y_i \leq \delta_i \Ytrain_i\;\forall i \in [n]}$.
    Then the density of the conditional distribution $Y_i \mid (\Ytrain, \delta)$ has the following form: 
    \begin{equation}\label{equ:con_density_MVN}
        Y \mid (\Ytrain, \delta) \sim \Norm_{n, A}(\mu+c^{-1}\sigma^2\delta, \sigma^2 I_n),    
    \end{equation}
    where $\Norm_{n, A}$ denotes an $n$-dimensional multivariate normal distribution truncated to the set $A$. 
\end{lemma}
\begin{proof}
    The conditional density is as follows:
    \begin{align}\label{eq:density}
    f_{Y|(\Ytrain, \delta)}(Y)
    &\propto f_{\delta|Y, \Ytrain}(\delta) f_{\Ytrain | Y}(\Ytrain) f_Y(Y)\nonumber\\
    &\propto \II[Y \in A] f_{\Ytrain - Y | Y}(\Ytrain - Y) f_Y(Y)\nonumber\\
    &=\II[Y \in A] f_{\zeta}(\Ytrain - Y) f_Y(Y) \nonumber\\
    &= \II[Y \in A] \prod_{i=1}^n f_{\zeta_i}(\Ytrain_i - Y_i) f_Y(Y)\nonumber\\
    &\propto \II[Y \in A] \exp\cbr{-\frac{1}{2\sigma^2}(Y-\mu)^\top(Y-\mu) - c^{-1}\sum_{i=1}^n|\Ytrain_i-Y_i|}\nonumber\\
    &= \II[Y \in A] \exp\cbr{-\frac{1}{2\sigma^2}(Y-\mu)^\top(Y-\mu) + c^{-1}\delta^\top(Y - \Ytrain)}\nonumber\\
    &\propto \II[Y \in A]\exp\cbr{-\frac{1}{2\sigma^2}(Y-\mu - c^{-1}\sigma^2\delta)^\top(Y-\mu - c^{-1}\sigma^2\delta)}.
    \end{align}

    Thus the density of $Y|(\Ytrain, \delta)$ is that of multivariate normal distribution with independent marginals, truncated to $A$, which restricts the domain to the intersection of halfspaces defined by elements of $\Ytrain$.
\end{proof}

It remains to derive a confidence interval for $\mu_i$. Note that since $Y \mid (\Ytrain, \delta)$ is a truncated normal distribution with a diagonal covariance, the marginal distributions are also truncated normal distributions (note that this is not in general true for an arbitrary covariance). Thus the marginal distribution $Y_i \mid (\Ytrain, \delta) \sim \Norm_{A_i}\rbr{\mu_i + c^{-1}\sigma^2 \delta_i, \sigma^2}$, where $A_i := \cbr{y \in \RR^n: \delta_i y_i \leq \delta_i \Ytrain_i}.$ A valid post-selection confidence interval for $\mu_i$ follows from the CDF of this truncated normal distribution, yielding Proposition 2.

\subsection{Proof of Proposition 6}

First, without loss of generality, let $\mu_1 = \dots = \mu_n = 0$. Then
\begin{align*}
    \EE\sbr{Y_{\SelG(Y)} - \mu_{\SelG(Y)}}
    = \EE\sbr{Y_{\SelG(Y)}}
    = \EE\sbr{\EE\sbr{Y_i | i= \SelG(Y)}}.
\end{align*}
To bound the conditional expectation, let $X_i := Y_i + \zeta_i$, where $\zeta_i \iidsim \Norm(0, c^2\sigma^2)$ is the additive noise. Thus, $\SelG(Y) = \Sel(X)$. Let $Z \sim \Norm_n(0, I_n)$, with maximum order statistic $Z_{(1)} = Z_{\Sel(Z)}$. We can bound the conditional expectation
\begin{align*}
    \EE\sbr{Y_i | i= \SelG(Y)}
    &= \EE\sbr{X_i - \zeta_i | i= \Sel(X)}\\
    &= \EE\sbr{X_i| i= \Sel(X)} - \EE\sbr{\zeta_i | i= \Sel(X)}\\
    &\geq \EE\sbr{X_i| i= \Sel(X)} - \EE\sbr{\zeta_i | i= \Sel(\zeta)}\\
    &= \sigma\sqrt{1 + c^2} \EE\sbr{Z_i| i= \Sel(Z)} - \sigma c\EE\sbr{Z_i | i= \Sel(Z)}\\
    &= \sigma (\sqrt{1+c^2} - c) \EE[Z_{(1)}]\\
    &\geq \frac{\sigma}{2} (\sqrt{1+c^2} - c) \sqrt{\log n},
\end{align*}
where the final inequality comes from~\citet{kamath_bounds}. Taking the expectation on each side completes the proof.

\section{Vignette \#1 (non-randomized): The winner's curse}
\label{sec:vignette-1_non-randomized}

In Vignette \#1 of the main text, we considered solutions to the ``winner's curse''~\citep{andrews_inference_2024} that made use of randomized selection algorithms. Yet there are also reasons to desire a non-randomized selection algorithm, and inference in that setting is a well-studied problem. We review this setting and then empirically compare existing approaches for inference. The conclusions mirror those in the main text.

Suppose again that we observe $n$ independent Gaussian random variables $Y_1,\ldots,Y_n$ with common known variance, and are interested in conducting inference on the unknown mean of the largest observed variable. That is, 
\begin{equation}\label{eq:vignette-1}
    Y \sim \Norm_n(\mu, \sigma^2 I_n), 
\end{equation}
    the selection rule is 
\begin{equation}
 \Sel(Y) := \argmax_{i \in [n]} Y_i,
 \label{eq:vignette-1-selection}
\end{equation}
and our target of inference is $\mu_{\Sel(Y)}$.

Research on how to construct valid confidence intervals for $\mu_{\Sel(Y)}$ is a long and well-studied problem. We do not claim to consider every applicable method, and will not need to do so in order to make our point clear. There are two main categories of methods. The first category guarantees unconditional coverage, defined in (1.2), via the stronger goal of simultaneous coverage. That is, they achieve
\begin{equation*}
    \Pr\rbr{\forall i \in \Ical : \theta_{i} \in \CI_{i}^{H(\alpha, \Scal, Y)}(Y)} \geq 1 - \alpha,
\end{equation*}
which in turn implies (1.2).
\citet{andrews_inference_2024} achieve simultaneous coverage by using the $1-\alpha$ quantile of the maximum of $n$ standard normal random variables; they term this a \emph{projection} confidence interval. \citet{zrnic_locally_2024} developed the locally simultaneous inference approach, which focuses the simultaneous correction on only a subset of likely models, thereby often improving power; these two approaches yield infer-and-widen intervals centered at $Y_{\Sel(Y)}$. \citet{zrnic_flexible_2025} subsequently developed the Zoom correction approach in the context of the winner's curse to improve upon the locally simultaneous approach by using properties of $\Sel(Y)$; this in contrast yields \emph{asymmetric} intervals.
 
The second category of methods approaches unconditional coverage via the stronger goal of conditional coverage~\citep{fithian_optimal_2017}. \citet{andrews_inference_2024} show that the conditional distribution $Y_{\Sel(Y)} \mid \rbr{\Sel(Y) = i}$ takes the form of a truncated normal distribution, and use that distribution to conduct inference on $\mu_{\Sel(Y)}$. In order to improve power, \citet{andrews_inference_2024} additionally develop a hybrid method, which uses the conditional confidence interval in cases where there is strong signal, and otherwise resorts to simultaneous intervals; note that the hybrid method no longer provides conditional coverage, only unconditional coverage.

All of the aforementioned approaches yield confidence intervals for $\mu_{\Sel(Y)}$ that attain $1-\alpha$ unconditional coverage, (1.2). We should prefer the approach that yields narrower intervals. \cref{fig:vignette-1-app}(a) displays the ratio of average $95\%$ confidence interval widths of the locally simultaneous inference (LSI)~\citep{zrnic_locally_2024} and hybrid~\citep{andrews_inference_2024} methods across simulated settings under the model in \eqref{eq:vignette-1}. We fix $\mu_2 = \dots = \mu_n = 0$ and $\sigma^2 = 1$, and vary values of $n$ and the scaling of $\mu_1$. Across all settings, the LSI method has a wider average confidence interval width than the hybrid method, i.e., the ratio is greater than one.
In \cref{fig:vignette-1-app}(b, c), we fix $n=100$ and consider low and high amounts of signal $\mu_1$ in panels (b) and (c), respectively.
We numerically plot the ``oracle'' infer-and-widen interval: the narrowest possible interval centered at $Y_{\Sel(Y)}$ which contains $\mu_{\Sel(Y)}$ with probability $1-\alpha$. We compare the ``oracle'' interval to the conditional~\citep{andrews_inference_2024}, hybrid~\citep{andrews_inference_2024}, and Zoom~\citep{zrnic_flexible_2025} intervals. We furthermore display the average empirical widths and coverages of the $95\%$ intervals from the classic, simultaneous (SI)~\citep{andrews_inference_2024}, and LSI~\citep{zrnic_locally_2024} approaches. We can see that the hybrid method performs extremely well and that there appears to be limited value in exploring better infer-and-widen intervals.

\begin{figure}[!htb]
    \centering
    \includegraphics[width=0.75\linewidth]{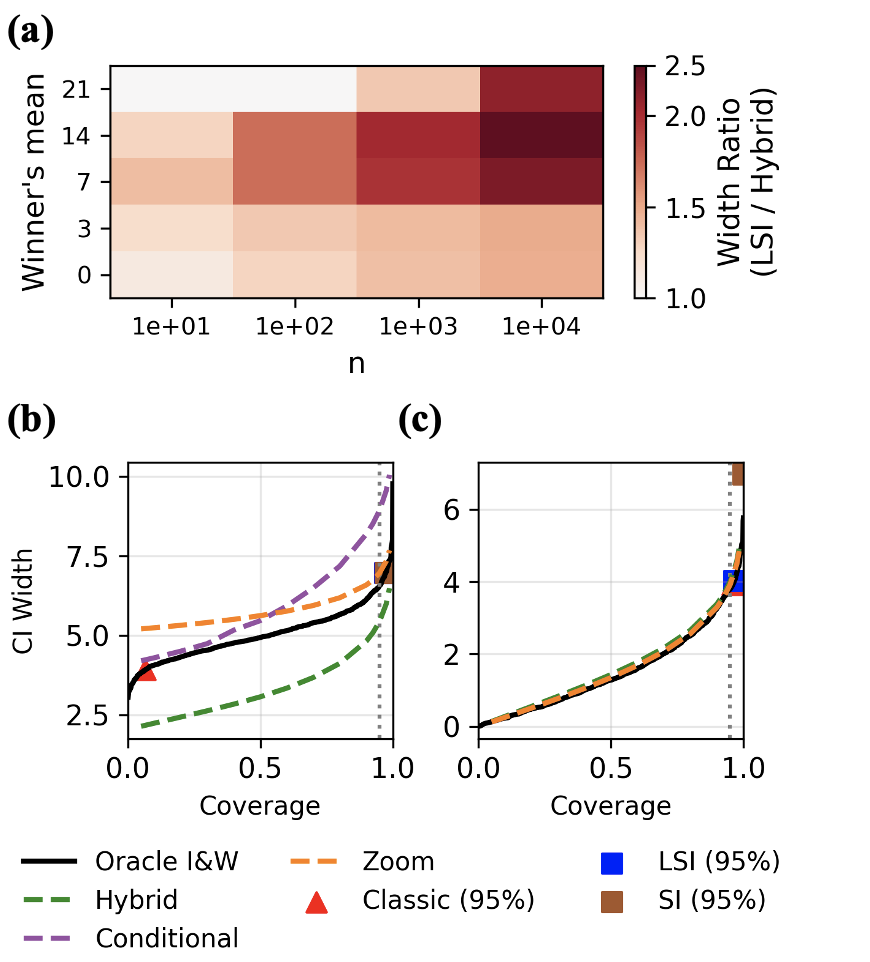}  
    \caption{Vignette \#1 under \eqref{eq:vignette-1} with  $\sigma=1$ and the selection rule in~\cref{eq:vignette-1-selection}. Results are averaged over (a) $500$ or (b, c) $1000$ simulated datasets. \textbf{(a):} We display the ratio of the average widths of $95\%$ confidence intervals from the LSI infer-and-widen method~\citep{zrnic_locally_2024} and the hybrid method~\citep{andrews_inference_2024}. A width ratio that exceeds one implies that the infer-and-widen interval is wider. \textbf{(b, c):} Fixing $n=100$, we compare the average width of the ``oracle'' infer-and-widen (I\textup{\symbol{`\&}}W) interval to the conditional~\citep{andrews_inference_2024}, hybrid~\citep{andrews_inference_2024}, and Zoom~\citep{zrnic_flexible_2025} intervals. For reference, we also display the average empirical widths and coverages of the $95\%$ intervals from the classic, simultaneous (SI)~\citep{andrews_inference_2024}, and LSI~\citep{zrnic_locally_2024} approaches. The signal $\mu_1$ is low and high in panels (b) and (c), respectively.
    }
    \label{fig:vignette-1-app}
\end{figure}
\clearpage
\section{Vignette \#2}

\subsection{Proof of Proposition 4}

\begin{proof}

   The proof will make use of Theorem 5.2 of \citet{lee_exact_2016}, which we restate now.

    \begin{theorem}[Theorem 5.2 of \citet{lee_exact_2016}]\label{thm:lee_taylor}
        Let $Y \sim \Norm_n(\mu, \Sigma)$ be an $n$-dimensional random vector,  and $\beta \in \RR^n$ a vector of interest. For any fixed matrix $D \in \RR^{n \times d}$ and vector $b \in \RR^d$, define
        \begin{equation*}
            \gamma := \Sigma \beta (\beta^\top \Sigma \beta)^{-1}
            \quad\mathrm{and}\quad
            W := (I_n - \gamma \beta^\top)Y,
        \end{equation*}
        and the pair
        \begin{equation*}
            v_{\min} := \underset{j : (D \gamma)_j < 0}{\mathrm{max}}(b_j - (D W)_j ) / (D \gamma)_j
            \quad\mathrm{and}\quad
            v_{\max} := \min_{j : (D \gamma)_j > 0}(b_j - (D W)_j ) / (D \gamma)_j.
        \end{equation*}
        Then
        \begin{equation}\label{eq:lee_taylor_tnorm}
            \beta^\top Y | \rbr{D Y \leq b, W} \sim \Norm_{[v_{\min}, v_{\max}]}(\beta^\top \mu, \beta^\top \Sigma \beta),
        \end{equation}
        where $\Norm_{[v_{\min}, v_{\max}]}$ denotes a normal distribution truncated to $[v_{\min}, v_{\max}]$.
    \end{theorem}

    Returning to our setting in Proposition 4, 
    for brevity let $Z := X^\top Y$ and $\Ztrain := Z + \zeta$ where $\zeta_1,\ldots,\zeta_p \iidsim \Lap(c)$. Recall that our target of inference is $X_{\SelLtwo(Y)}^\top \mu$, where $\SelLtwo(Y)$ is the selection rule in (4.2). 
    This selection rule is implicitly a function of the Laplacian noise $\zeta$, and so to emphasize this point we will write $\SelLtwo(Y; \zeta)$ in this section.

    To construct a confidence interval with conditional coverage as in (1.3), 
    it will suffice to derive the distribution of $X_{\SelLtwo(Y; \zeta)}^\top Y \mid C(Y, \zeta)$ for an event $\Ccal(Y)$ satisfying $\sigma\rbr{C(Y, \zeta)} \supseteq \sigma \rbr{\SelLtwo(Y; \zeta)}$, where $\sigma(\cdot)$ denotes the sigma algebra generated by a random variable, e.g., we will condition on  the selection event along with an additional event.

    We will show that there exists a matrix $A_{\SelLtwo(y; \zeta^{obs})} \in \RR^{(2p + 1) \times p}$ for which
    \begin{equation}\label{eq:equivalence}
        \cbr{(Y, \zeta): \SelLtwo(Y; \zeta) = \SelLtwo(y; \zeta^{obs})} = \cbr{(Y, \zeta): \delta A_{\SelLtwo(y; \zeta^{obs})} X^\top Y \leq -\delta A_{\SelLtwo(y; \zeta^{obs})} \zeta},    
    \end{equation}
    where 
    \begin{equation}
        \delta := \sign\rbr{X_{\SelLtwo(y;\zeta^{obs})}^\top Y + \zeta_{\SelLtwo(y;\zeta^{obs})}} = \sign\rbr{\Ztrain_{\SelLtwo(y;\zeta^{obs})}}.
        \label{eq:delta}
    \end{equation}
    Henceforth, for brevity, we will denote $A_{\SelLtwo(y; \zeta^{obs})}$ by $A$.

    To show that there exists a matrix $A$ for which \eqref{eq:equivalence} holds, note that by definition $\SelLtwo(Y; \zeta) = \SelLtwo(y; \zeta^{obs})$ if and only if $|\Ztrain_{\SelLtwo(y; \zeta^{obs})}| \geq |\Ztrain_j|$ for all $j \in [p]$. In turn, expanding the absolute value and definition of $\Ztrain$, \eqref{eq:equivalence}   holds if and only if both (I) and (II) hold: 
    \begin{enumerate}[(I)]
        \item if $\Ztrain_{\SelLtwo(y; \zeta^{obs})} \geq 0$, then for all $j \neq \SelLtwo(y; \zeta^{obs})$ it holds that
        \begin{enumerate}
            \item $Z_{\SelLtwo(y; \zeta^{obs})} + \zeta_{\SelLtwo(y; \zeta^{obs})} \geq 0$,
            \item $Z_{\SelLtwo(y; \zeta^{obs})} + \zeta_{\SelLtwo(y; \zeta^{obs})} \geq Z_{j} + \zeta_{j}$,
            \item $Z_{\SelLtwo(y; \zeta^{obs})} + \zeta_{\SelLtwo(y; \zeta^{obs})} \geq -(Z_{j} + \zeta_{j})$;
        \end{enumerate}
        \item else if $\Ztrain_{\SelLtwo(y; \zeta^{obs})} \leq 0$, then for all $j \neq \SelLtwo(y; \zeta^{obs})$ it holds that 
        \begin{enumerate}
            \item $Z_{\SelLtwo(y; \zeta^{obs})} + \zeta_{\SelLtwo(y; \zeta^{obs})} \leq 0$,
            \item $Z_{\SelLtwo(y; \zeta^{obs})} + \zeta_{\SelLtwo(y; \zeta^{obs})} \leq Z_{j} + \zeta_{j}$,
            \item $Z_{\SelLtwo(y; \zeta^{obs})} + \zeta_{\SelLtwo(y; \zeta^{obs})} \leq -(Z_{j} + \zeta_{j})$.
        \end{enumerate}
    \end{enumerate}

    We further rearrange the above inequalities to collect the $Z$ terms on one side, as follows:

    \begin{enumerate}[(I)]
        \item If $\Ztrain_{\SelLtwo(y; \zeta^{obs})} \geq 0$, then for all $j \neq \SelLtwo(y; \zeta^{obs})$ it holds that
        \begin{enumerate}
            \item $-Z_{\SelLtwo(y; \zeta^{obs})} \leq \zeta_{\SelLtwo(y; \zeta^{obs})}$,
            \item $Z_j - Z_{\SelLtwo(y; \zeta^{obs})} \leq \zeta_{\SelLtwo(y; \zeta^{obs})} - \zeta_j$,
            \item $-Z_j - Z_{\SelLtwo(y; \zeta^{obs})} \leq \zeta_{\SelLtwo(y; \zeta^{obs})} + \zeta_j$;
        \end{enumerate}
        \item else if $\Ztrain_{\SelLtwo(y; \zeta^{obs})} \leq 0$, then for all $j \neq \SelLtwo(y; \zeta^{obs})$ it holds that 
        \begin{enumerate}
            \item $Z_{\SelLtwo(y; \zeta^{obs})} \leq -\zeta_{\SelLtwo(y; \zeta^{obs})}$,
            \item $Z_{\SelLtwo(y; \zeta^{obs})} - Z_{j} \leq \zeta_{j} - \zeta_{\SelLtwo(y; \zeta^{obs})}$,
            \item $Z_j + Z_{\SelLtwo(y; \zeta^{obs})} \leq -\zeta_{\SelLtwo(y; \zeta^{obs})} - \zeta_j$.
        \end{enumerate}
    \end{enumerate}

    Thus, if $\Ztrain_{\SelLtwo(y; \zeta^{obs})} \geq 0$, then we are in case (I), and $\SelLtwo(Y; \zeta) = \SelLtwo(y; \zeta^{obs})$ if and only if $Z$ satisfies $A Z \leq -A \zeta$, where
    \begin{align}
        A := \begin{bmatrix}
            0_{1 \times n_1} & -1 & 0_{1 \times n_2} \\
            I_{n_1} & -1_{n_1} & 0_{n_1 \times n_2} \\
            0_{n_2 \times n_1} & -1_{n_2} & I_{n_2} \\
            -I_{n_1} & -1_{n_1} & 0_{n_1 \times n_2} \\
            0_{n_2 \times n_1} & -1_{n_2} & -I_{n_2} 
        \end{bmatrix}, \label{eq:A}
    \end{align}
    and we denote $n_1 := \SelLtwo(y; \zeta^{obs})-1$ and $n_2 := p - \SelLtwo(y; \zeta^{obs})$ for brevity. If instead $\Ztrain_{\SelLtwo(y; \zeta^{obs})} \leq 0$, then we are in case (II), where $\SelLtwo(Y; \zeta) = \SelLtwo(y; \zeta^{obs})$ if and only if $Z$ satisfies $-A Z \leq A \zeta$. Thus, in either case,  $\SelLtwo(Y; \zeta) = \SelLtwo(y; \zeta^{obs})$ if and only if $\delta A Z \leq - \delta A \zeta$, and so, by definition of $Z$, if and only if $\delta A X^\top Y \leq - \delta A \zeta$. This establishes that \eqref{eq:equivalence} holds for $A$ defined in \eqref{eq:A}.

    Thus, by~\cref{eq:equivalence}, we see that
    \begin{equation*}
        X_{\SelLtwo(Y; \zeta)}^\top Y | \rbr{\SelLtwo(Y; \zeta) = \SelLtwo(y; \zeta^{obs})}
        \eqindist
        X_{\SelLtwo(Y; \zeta)}^\top Y | \rbr{\delta A X^\top Y \leq -\delta A \zeta}.
    \end{equation*}
    We would like to apply~\cref{thm:lee_taylor} to the right-hand side of this result with $D = \delta A X^\top$ and $b = -\delta A \zeta$, but these are random quantities and so we will need to fix them by conditioning on $\zeta$ and $\delta$.

    Defining $W=(I_n - X_{\SelLtwo(y; \zeta^{obs})} X_{\SelLtwo(y; \zeta^{obs})}^\top) Y$ as in~\cref{thm:lee_taylor}, we have that 
    \begin{align}
        &X_{\SelLtwo(Y; \zeta)}^\top Y | \rbr{\SelLtwo(Y; \zeta) = \SelLtwo(y; \zeta^{obs}), \zeta = \zeta^{obs}, \delta = \delta^{obs}, W}\label{eq:from_selection}\\
        &\eqindist X_{\SelLtwo(y; \zeta^{obs})}^\top Y | \rbr{\SelLtwo(Y; \zeta) = \SelLtwo(y; \zeta^{obs}), \zeta = \zeta^{obs}, \delta = \delta^{obs}, W}\nonumber\\
        &\eqindist
        X_{\SelLtwo(y; \zeta^{obs})}^\top Y | \rbr{\delta A X^\top Y \leq -\delta A \zeta, \zeta = \zeta^{obs}, \delta = \delta^{obs}, W}\label{eq:to_lineqs}\\
        &\eqindist
        X_{\SelLtwo(y; \zeta^{obs})}^\top Y | \rbr{\delta^{obs} A X^\top Y \leq -\delta^{obs} A \zeta, \zeta = \zeta^{obs}, \delta = \delta^{obs}, W} \nonumber\\
        &\eqindist
        X_{\SelLtwo(y; \zeta^{obs})}^\top Y | \rbr{\delta^{obs} A X^\top Y \leq -\delta^{obs} A \zeta, \zeta = \zeta^{obs}, W}\label{eq:drop_delta}\\
        &\eqindist
        X_{\SelLtwo(y; \zeta^{obs})}^\top Y | \rbr{\delta^{obs} A X^\top Y \leq -\delta^{obs} A \zeta^{obs}, \zeta = \zeta^{obs}, W}\nonumber  \\
        &\eqindist
        X_{\SelLtwo(y; \zeta^{obs})}^\top Y | \rbr{\delta^{obs} A X^\top Y \leq -\delta^{obs} A \zeta^{obs}, W}\label{eq:drop_zeta}.
    \end{align}

    \cref{eq:to_lineqs} follows from the equivalence in~\cref{eq:equivalence}.  
    
    To establish \cref{eq:drop_delta}, note  that the event $\{ \delta^{obs} A X^\top Y \leq -\delta^{obs} A \zeta \}$ is composed of a system of $2p+1$ inequalities, the first of which states that $ -\delta^{obs} X_{\SelLtwo(y; \zeta^{obs})}^\top Y \leq \delta^{obs} \zeta_{\SelLtwo(y; \zeta^{obs})}$, or equivalently, that 
    $0 \leq   \delta^{obs} \left(  X_{\SelLtwo(y; \zeta^{obs})}^\top Y + \zeta_{\SelLtwo(y; \zeta^{obs})} \right)$. Recalling the definition of $\delta$ in \eqref{eq:delta}, it follows that the event $\{ \delta^{obs} A X^\top Y \leq -\delta^{obs} A \zeta \}$   implies that $\delta=\delta^{obs}$. This yields 
    \eqref{eq:drop_delta}. 

    Finally, recall that $\zeta \independent Y$.  Note that $X_{\SelLtwo(y; \zeta^{obs})}^\top Y$, $W$, and the event $\rbr{\delta^{obs} A X^\top Y \leq -\delta^{obs} A \zeta^{obs} }$ are functions of $Y$, and thus all three are independent of $\zeta$. \cref{eq:drop_zeta}  follows. 

    We now apply \cref{thm:lee_taylor} with $Y \sim \Norm_n(\mu, \sigma^2 I_n)$, 
    $\beta = X_{\SelLtwo(y; \zeta^{obs})}$, $D = \delta^{obs} A X^\top$, and $b=-\delta^{obs} A \zeta^{obs}$. 
    Furthermore, recalling that $X$ is column-normalized, we have that $\gamma= X_{\SelLtwo(Y; \zeta^{obs})}$ and $W=\left(I_n - X_{\SelLtwo(y; \zeta^{obs})} X_{\SelLtwo(y; \zeta^{obs})}^\top \right) Y$ in the statement of \cref{thm:lee_taylor}. Thus, 
    \begin{equation}\label{eq:lee_dis}
        X_{\SelLtwo(y; \zeta^{obs})}^\top Y | \rbr{\delta^{obs} A X^\top Y \leq -\delta^{obs} A \zeta^{obs}, W}
        \sim \Norm_{[v_{\min}, v_{\max}]}(X_{\SelLtwo(y; \zeta^{obs})}^\top \mu, \sigma^2), 
    \end{equation}
    for $v_{\min}$ and $v_{\max}$ defined by
    \begin{align*}
        v_j &:= -(A\zeta + A X^\top W)_j / (A X^\top X_{\SelLtwo(Y)})_j,\\
        v_{\max} &:= \underset{j : \rbr{\delta A X^\top X_{\SelLtwo(Y)}}_j < 0}{\mathrm{max}} v_j,\\
        v_{\min} &:= \min_{j : \rbr{\delta A X^\top X_{\SelLtwo(Y)}}_j > 0} v_j.
    \end{align*}
    
   Applying \cref{eq:from_selection}-\cref{eq:drop_zeta} to \eqref{eq:lee_dis}, we have that
    \begin{equation*}
        X_{\SelLtwo(Y; \zeta)}^\top Y | \rbr{\SelLtwo(Y; \zeta) = \SelLtwo(y; \zeta^{obs}), \zeta=\zeta^{obs}, \delta=\delta^{obs}, W}
        \sim \Norm_{[v_{\min}, v_{\max}]}(X_{\SelLtwo(y; \zeta^{obs})}^\top \mu, \sigma^2).
    \end{equation*}
    
    A confidence interval for $X_{\SelLtwo(Y; \zeta)}^\top \mu ~| \rbr{\SelLtwo(Y; \zeta), \zeta, \delta, W}$ with conditional coverage guarantee (1.3), where $\Ccal(Y) = \rbr{\SelLtwo(Y), \zeta, \delta, W}$, is thus obtained using the CDF of the above truncated normal distribution.
\end{proof}

\subsection{Additional simulation results}

The simulation results in Section 4 
reflect a limited set of parameter choices. \cref{fig:vignette-2-width-ratio-grid} extends Figure 2(a) 
to a wider range of levels of signal. \cref{fig:vignette-2-oracle-grid} extends
Figure 2(b, c) 
to a wider range of levels of signal.

\begin{figure}[!htb]
\textbf{(a)} $\phi = 0$ \hspace{35mm} \textbf{(b)} $\EE[\phi] = 5$ \hspace{30mm} \textbf{(c)} $\EE[\phi] = \tfrac{50}{7}$
    \centering
    \includegraphics[width=\linewidth]{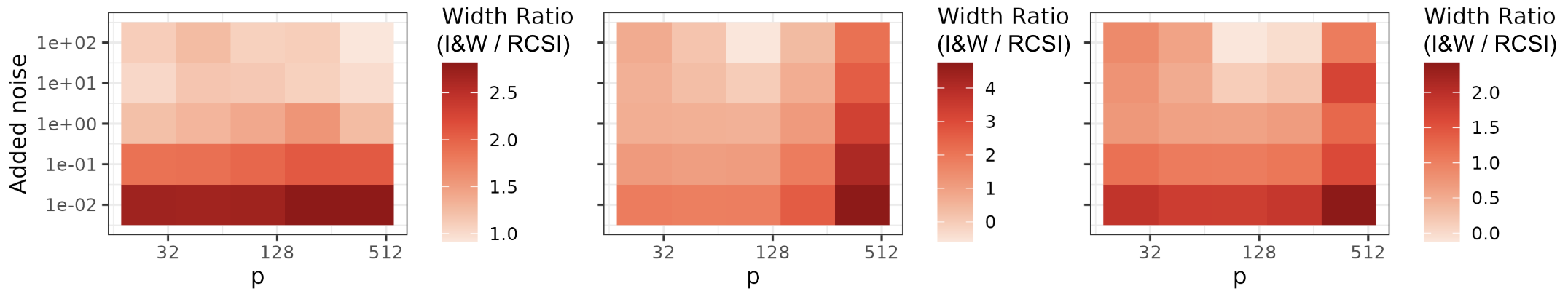}
    \caption{Vignette \#2, under the model (3.1) with $n=100$, $\sigma=1$, $\mu=X \phi$, and the selection rule in (4.2), where $X$ is multivariate normal with column correlation $0.5$, and $\phi$ is sparse with exponentially distributed non-zero elements. We vary the mean of $\phi$ across $0, 5, 50/7$ in panels (a), (b), and (c), respectively. Results are averaged over $250$ simulated datasets. We display the ratio of mean widths of infer-and-widen (4.3) and randomized conditional selective inference (RCSI) (4.6) intervals with $95\%$ coverage across the number of features $p$ and Laplacian noise variance $2c^2$. A width ratio that exceeds one implies that the average infer-and-widen interval is wider. The differences in widths is most pronounced under the null in panel (a), but results are qualitatively similar across levels of signal.}
    \label{fig:vignette-2-width-ratio-grid}
\end{figure}

\begin{figure}[!htb]
\hspace{20mm}$\phi = 0$ \hspace{34mm} $\EE[\phi] = 5$ \hspace{34mm} $\EE[\phi] = \tfrac{50}{7}$
    \centering
    \includegraphics[width=0.9\linewidth]{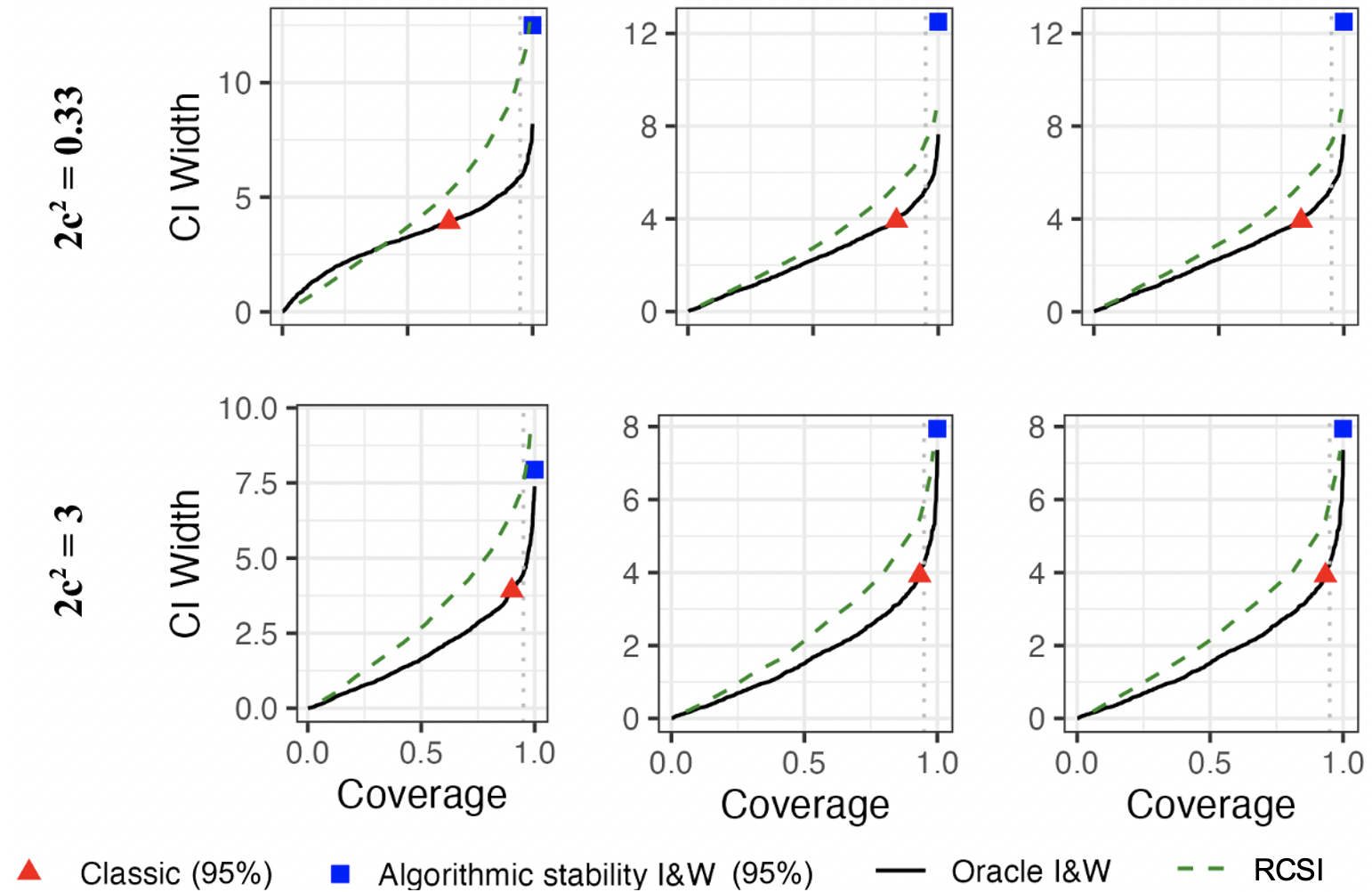}
    \caption{Vignette \#2, under the model (3.1) with $n=100$, $p=100$, $\sigma=1$, $\mu= \phi$, and the selection rule in (4.2), where $X$ is multivariate normal with column correlation $0.5$, and $\phi$ is sparse with exponentially distributed non-zero elements. We vary the mean of $\phi$ across the three columns, and the Laplacian noise variance across the two rows. We display the average width of the ``oracle'' infer-and-widen interval and the randomized conditional selective inference (RCSI) interval (4.6) as a function of coverage. The 95\% classical interval, and the ``attainable'' 95\% infer-and-widen interval (4.3) based on algorithmic stability, are also displayed. Results are qualitatively similar regardless of choice of parameters.}
    \label{fig:vignette-2-oracle-grid}
\end{figure}

\clearpage

\section{Vignette \#3 (randomized): Inference after the lasso}\label{sec:vignette-3_randomized}

We now consider inference on the set of coefficients selected by the lasso in a linear model \citep{tibshirani1996regression}. We again consider the model~\cref{eq:vignette-1}, alongside a fixed design matrix $X \in \RR^{n \times p}$ with unit-normalized columns, i.e., $\nbr{X_1}_2 = \dots = \nbr{X_p}_2 = 1$. Our new selection rule is the magnitude-ordered set of features selected by the lasso:
\begin{align}
\label{eq:vignette-3-selection}
\Selthree(Y) = &\bigg\{ j_1,\dots,j_l: \abr{\hat\beta_{j_1}^\lambda} > \dots > \abr{\hat\beta_{j_l}^\lambda} > 0, \; \\&\quad
\underset{k \notin \{ j_1, \ldots, j_l\}}{\max} |\hat\beta_k^\lambda | = 0\bigg\}, \nonumber
\end{align}
where
\begin{equation*}
    \hat\beta^\lambda := \argmin_\beta \left\{ \nbr{ Y - X \beta }^2_2 + \lambda \|  \beta \|_1 \right\}\nonumber.
\end{equation*}
We then consider a least squares model fit to the features in $\Selthree(Y)$. 
Our target of inference is the coefficient of the feature with largest lasso magnitude, i.e., $\beta_{{\Selthree(Y)},1}$,  where
\begin{equation*}
    \beta_{\Selthree(Y)}:=  A_{\Selthree(Y)} \EE\sbr{Y}
\end{equation*}
and
\begin{equation*}
    A_{J} := \rbr{X_{J}^\top X_{J}}^{-1} X_{J}^\top
\end{equation*}
for any subset of features $J \subseteq [p]$. 

The lasso is not stable in the sense of~\citep{zrnic_post-selection_2023}. Therefore, \cite{zrnic_post-selection_2023} consider a variant of the lasso that is far enough from the original lasso proposal that we do not investigate it here. 
Instead, in what follows, we construct an interval based on Gaussian data thinning \citep{neufeld_data_2024,dharamshi_generalized_2024} that yields conditional guarantees in the sense of (1.3), and then consider the ``oracle'' infer-and-widen interval that makes use of the identical selection event.

Towards this end, for some $c > 0$, we define
\begin{align}
\label{eq:vignette-3-selection-gaussian}
&\SelthreeG(Y) = \Selthree(Y + c\zeta),
\quad \zeta_i \iidsim \Norm(0, \sigma^2).
\end{align}

\begin{proposition}\label{thm:vignette-3-gaussian-DT}
Under \eqref{eq:vignette-1}, consider the selection rule $\SelthreeG(\cdot)$ defined in \eqref{eq:vignette-3-selection-gaussian}.
The data thinning confidence interval
\begin{equation}\label{eq:vignette-3-CI-GDT}
    \CI_{\SelthreeG(Y)}^{\alpha}(Y)
   :=  \sbr{ \hat\beta_{\SelthreeG(Y), 1} \pm z_{1-\alpha/2}\rbr{\Sigma_{1,1}}^{1/2}  }
\end{equation}
has $1-\alpha$ conditional coverage for $\beta_{\SelthreeG(Y), 1}$ in the sense of~(1.3) with $\Ccal(Y) := \rbr{Y + c\zeta},$
where
$$\hat\beta_{\SelthreeG(Y)} := A_{\SelthreeG(Y)} \rbr{Y - \tfrac1c \zeta}$$ 
and
$$\Sigma := \sigma^2 \rbr{1 + \tfrac{1}{c^2}} \rbr{X_{\SelthreeG(Y)}^\top X_{\SelthreeG(Y)}}^{-1}.$$
\end{proposition}

We compare the widths of the data thinning (DT)~\eqref{eq:vignette-3-CI-GDT} and ``oracle'' infer-and-widen intervals ---  with identical selection events --- in Figure~\ref{fig:oracle_lasso}. 
We simulate data according to \cref{eq:vignette-1} with $n=100$ and $\sigma=1$. We set $\mu = X \phi$, where half of the elements of $\phi$ are zero and the other half are either zero or drawn from exponential distributions with means $50/7$ or $5$ in panels (a), (b), and (c) of~\cref{fig:oracle_lasso}. The rows of $X$ are independent draws from $\Norm_p(0, \rho (1_p1_p^\top  - I_p) + I_p)$ with $\rho=0.5$, and the columns of $X$ are scaled so that $\nbr{X_1}_2 = \dots = \nbr{X_p}_2 = 1$. The penalty $\lambda$ is adaptively chosen using 3-fold cross-validation. \cref{fig:oracle_lasso} displays the ratio of the average confidence interval widths as we vary $p$ and the variance of the Gaussian noise $c^2$ in \eqref{eq:vignette-3-selection-gaussian}.

We see from Figure~\ref{fig:oracle_lasso} that the relative widths of the two intervals depend on the variance of the noise in \eqref{eq:vignette-3-selection-gaussian} and distribution of $\phi$.
Of course, the ``oracle'' infer-and-widen interval cannot be computed in practice: any real-world infer-and-widen interval that achieves nominal coverage would necessarily be at least as wide as the oracle interval, if not far wider.
Thus, while it may be possible to develop an infer-and-widen interval that might outperform the simple data thinning interval in~\cref{eq:vignette-3-CI-GDT} in specific scenarios, the potential upside of such an approach may be limited.

\begin{figure}[!htb]
    \centering
    \includegraphics[width=0.4\linewidth]{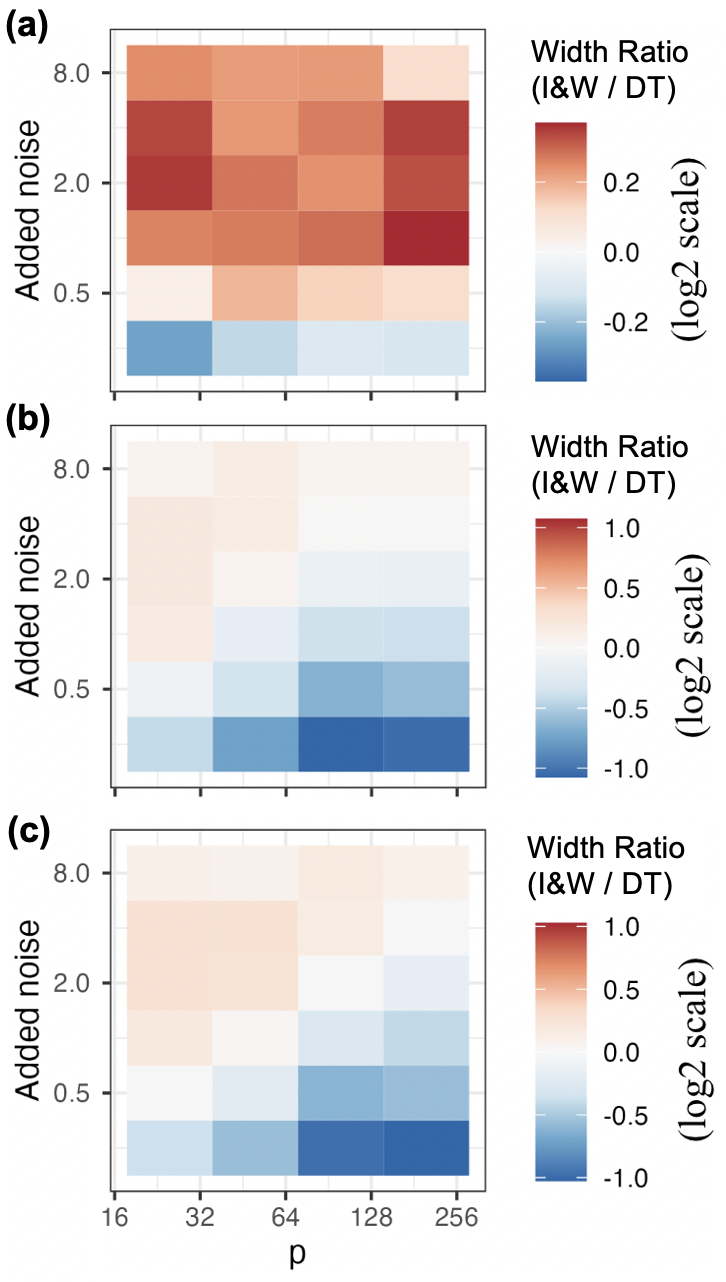} 
    \caption{Vignette \#3 under \eqref{eq:vignette-1} with $n=100$, $\sigma=1$, $\mu=X \phi$, and the selection rule in~\cref{eq:vignette-3-selection-gaussian}. The columns of $X$ are independent normal draws with column correlation $0.5$. Results are averaged over $1000$ simulated datasets. We consider inference on a lasso-selected coefficient in a least squares model, varying the mean of $\phi$ across $0, 5, 50/7$ in panels (a), (b), and (c), respectively. We display the ratio of average widths of the 95\% ``oracle'' infer-and-widen interval and the 95\% data thinning (DT) interval~\eqref{eq:vignette-3-CI-GDT}, with identical selection events and level $\alpha=0.05$; a positive log ratio indicates that the data thinning interval is narrower and hence preferable. In many settings, even the unattainable ``oracle'' infer-and-widen interval is wider than the attainable data thinning interval.
    }
    \label{fig:oracle_lasso}
\end{figure}

\subsection{Proof of Proposition 5}

The result follows from independence of the selection and inference stages, per the following result.

\begin{lemma}[Gaussian data thinning, \citet{rasines_splitting_2023, neufeld_data_2024}]\label{prop:thinning}
    Assume $Y \sim \Norm_n(\mu, \Sigma)$, where $\mu$ is unknown and $\Sigma$ is known. For any $c > 0$, let $\zeta \sim \Norm_n(0, \Sigma)$, and define
        $\Ytrain := Y + c\zeta$ and    $\Ytest := Y - \tfrac1c \zeta$.
    Then 
        $\Ytrain \sim \Norm_n\rbr{\mu, \rbr{1 + c^2}\Sigma}$, $  \Ytest \sim \Norm_n\rbr{\mu, \rbr{1 + \tfrac{1}{c^2}}\Sigma}$, and $\Ytrain \independent \Ytest$.
\end{lemma}

The selection event $\SelthreeG(Y)$ in (5.3) 
is a function of $Y$ through $\Ytrain$ and is independent of $\Ytest$. We proceed to conduct conventional inference using $\Ytest$ in an ordinary least squares model on the coefficient selected using $\Ytrain$, as given in Section 4.1 of ~\citet{berk_valid_2013}. This yields a conditional guarantee of the form in (1.3), with $\Ccal(Y) = (\Ytrain).$

\bibliographystyle{plainnat}
\bibliography{references}